\documentclass[sigconf]{acmart}

\AtBeginDocument{%
  }

\copyrightyear{2026}
\acmYear{2026}
\setcopyright{cc}
\setcctype{by}
\acmConference[SIGIR '26]{Proceedings of the 49th International ACM SIGIR Conference on Research and Development in Information Retrieval}{July 20--24, 2026}{Melbourne, VIC, Australia}
\acmBooktitle{Proceedings of the 49th International ACM SIGIR Conference on Research and Development in Information Retrieval (SIGIR '26), July 20--24, 2026, Melbourne, VIC, Australia}
\acmDOI{10.1145/3805712.3809597}
\acmISBN{979-8-4007-2599-9/2026/07}

\usepackage{multirow}
\usepackage{balance}

\begin{document}

\title{Post-hoc Provider Fairness Adaptation via Hierarchical Exposure Alignment}

\author{Jingzhi Li}
\affiliation{%
  \institution{The School of Computer Science and Information Engineering, Hefei University of Technology,}
  \city{Hefei, Anhui}
  \country{China}}
\email{jingzhili826@gmail.com}

\author{Zhiyong Cheng}
\authornote{Corresponding author.}
\affiliation{%
  \institution{The School of Computer Science and Information Engineering, Hefei University of Technology}
  \city{Hefei, Anhui}
  \country{China}}
\email{jason.zy.cheng@gmail.com}

\author{Richang Hong}
\authornotemark[1]
\affiliation{%
  \institution{The School of Computer Science and Information Engineering, Hefei University of Technology}
  \city{Hefei, Anhui}
  \country{China}}
\email{hongrc.hfut@gmail.com}

\author{Meng Wang}
\affiliation{%
  \institution{The School of Computer Science and Information Engineering, Hefei University of Technology}
  \city{Hefei, Anhui}
  \country{China}}
\email{eric.mengwang@gmail.com}
\renewcommand{\shortauthors}{Jingzhi Li, Zhiyong Cheng, Richang Hong, and Meng Wang}
\begin{abstract}
Provider exposure fairness is crucial for sustaining a healthy content ecosystem and preventing monopolization in recommender systems. Yet, most existing methods either \textit{incorporate fairness constraints during model training}, requiring expensive retraining when fairness objectives change, or rely on \textit{post-hoc reranking with fixed criteria}, which lacks adaptability to diverse fairness requirements. To overcome these limitations, we propose \textbf{Post-hoc Fairness Adaptation} (PFA), a lightweight framework that equips a frozen recommender with a fairness adapter, enabling flexible fairness control without retraining the backbone model. Specifically, the fairness adapter learns personalized additive score adjustments from user–item embeddings, which are injected into the original ranking scores to steer provider exposure toward fairness. To train the adapter, we minimize the Kullback-Leibler (KL) divergence between the actual and the target fair exposure distributions. However, this global objective implicitly treats all providers equally, ignoring structural disparities such as imbalanced provider group sizes and heterogeneous exposure within groups. Consequently, fairness may appear satisfied at an aggregate level while severe inter-group and intra-group exposure imbalances persist, undermining practical fairness. To address this, we design \textbf{Hierarchical Exposure Fairness Alignment} (HEFA), which explicitly balances inter- and intra-group provider exposure disparities, enabling flexible adaptation to diverse fairness requirements. To mitigate potential accuracy degradation, PFA jointly optimizes HEFA with a differentiable NDCG loss, enabling end-to-end fairness optimization while preserving ranking quality. Extensive experiments on three public datasets demonstrate that PFA achieves substantial fairness gains with negligible accuracy loss, consistently outperforming strong baselines. Code is available at \url{https://github.com/Tam-JQK/Post-train}. 
\end{abstract}

\maketitle

\section{Introduction}
Recommender systems play a central role in connecting users with relevant content on modern online platforms~\cite{Survey1,Survey2,new7,CARD}. While predictive accuracy has long been the primary optimization objective, optimizing solely for accuracy often leads to exposure concentration on a small subset of head providers~\cite{new2,new3,new4}. Such skewed traffic allocation reinforces a "\textit{winner-takes-all}" dynamic, suppressing the visibility of long-tail providers and eroding ecosystem diversity, which in turn harms long-term user engagement~\cite{Survey4,Survey5,sssurney}. Consequently, ensuring fair provider exposure is no longer optional but a core design principle for sustainable recommendation ecosystems.

Existing approaches to provider exposure fairness fall broadly into two categories, as illustrated in Figure~\ref{fig:motivation}. \textit{In-processing} methods~\cite{FairNeg,Multi-FR,FairDual} integrate fairness constraints directly into model training through techniques such as resampling~\cite{FairNeg}, reweighting~\cite{FairDual}, or fairness-aware regularization~\cite{Multi-FR}. While these methods can jointly optimize accuracy and fairness, they tightly couple fairness objectives with model parameters, making costly retraining inevitable when fairness requirements change. \textit{Post-processing} methods~\cite{FairRec,p-mmf,CPFair,TFROM} instead rerank the outputs of a pretrained recommender, making them appealing for their model-agnostic nature. Representative approaches include fair allocation mechanisms~\cite{FairRec}, combinatorial optimization formulations~\cite{CPFair}, and constrained reranking strategies~\cite{p-mmf}. However, these methods rely on fixed fairness criteria embedded in reranking algorithms, limiting their adaptability to diverse and evolving fairness requirements. These limitations highlight the need for an adaptive framework that decouples fairness optimization from base recommender training. Recent advances in parameter-efficient fine-tuning, particularly adapter-based methods~\cite{fine-tuning1,fine-tuning3,fine-tuning2}, show that lightweight modules can effectively adapt pretrained models to new objectives without modifying backbone parameters. Inspired by this, we ask: \textit{Can we achieve flexible fairness adaptation for pretrained recommenders through lightweight, post-hoc modules, without resorting to full retraining or rigid reranking strategies?}
\begin{figure}[htbp]
    \centering
    \includegraphics[width=\columnwidth]{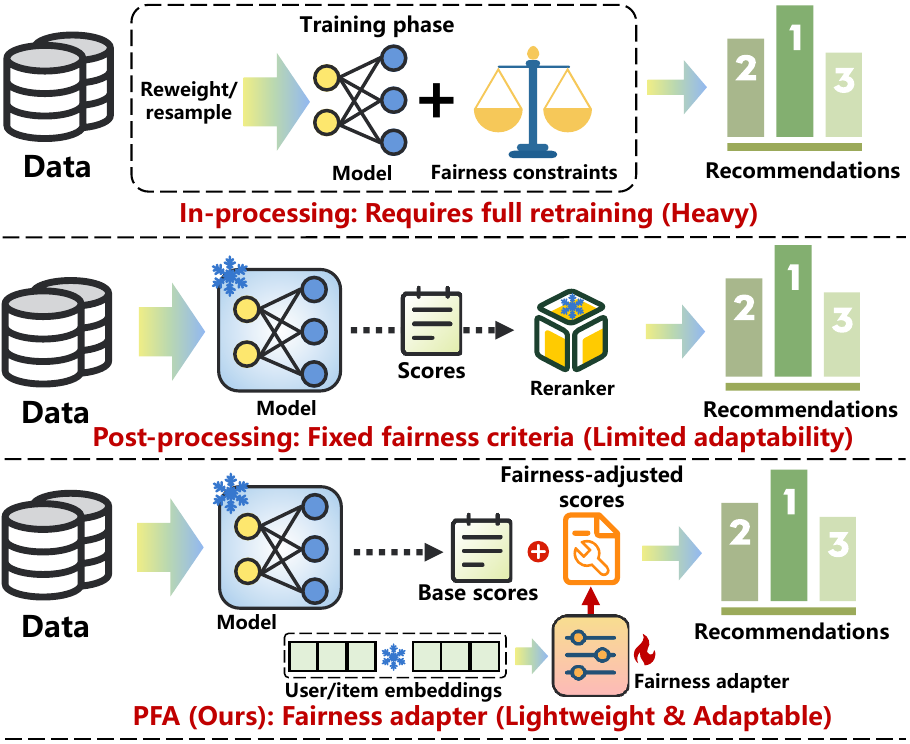}
    \caption{Comparison of in-processing, post-processing, and PFA (ours) for provider exposure fairness. PFA achieves fairness optimization by training a lightweight fairness adapter.}
    \label{fig:motivation}
\end{figure}

To address this question, we propose PFA, a framework that equips a frozen pretrained recommender with a lightweight fairness adapter. The adapter operates on user–item embeddings to learn personalized additive score adjustments, which are injected into the original ranking scores to steer recommendations toward fairer exposure distributions. Under this design, a natural training objective is to align the achieved provider exposure distribution with a target fair distribution by minimizing the KL divergence. However, this global alignment objective implicitly treats all providers equally, overlooking structural disparities such as imbalanced provider group sizes and heterogeneous exposure within groups. Consequently, fairness may appear satisfied at an aggregate level while substantial inter-group and intra-group exposure imbalances persist, limiting practical effectiveness. To overcome this, we propose HEFA, a flexible fairness objective that hierarchically models provider exposure. By decomposing fairness optimization into complementary inter-group and intra-group components, HEFA enables independent control of exposure balance at different granularities, thereby allowing adaptation to diverse fairness requirements. Since fairness-oriented score corrections may degrade ranking accuracy, PFA jointly optimizes the HEFA loss with a differentiable NDCG loss to preserve recommendation quality. To support end-to-end optimization, we further adopt the differentiable sorting network~\cite{DSN} to approximate the non-differentiable ranking and exposure computation operations. 

Our main contributions are summarized as follows:
\begin{itemize}
\item We propose \textbf{Post-hoc Fairness Adaptation} (PFA), a lightweight framework that improves provider exposure fairness via a trainable adapter while keeping the base recommender frozen, thereby avoiding costly retraining of the backbone.
\item We design \textbf{Hierarchical Exposure Fairness Alignment} (HEFA), a flexible fairness objective that decomposes exposure fairness into inter-group and intra-group components, enabling fine-grained and independent control of fairness optimization at different granularities.
\item Extensive experiments on three public datasets demonstrate that PFA consistently improves provider exposure fairness while maintaining competitive recommendation accuracy.
\end{itemize}

\section{Related Work}
\label{Related_work}
Recommender systems have become a core infrastructure of online platforms, and algorithmic fairness has emerged as an important consideration in recent recommender systems research~\cite{consumer1,provider1,fairnessdef,new1}. Fairness in recommender systems primarily concerns two stakeholder groups: consumers and providers. \textit{Consumer-side fairness} typically focuses on reducing performance disparities across user groups, for example, by enforcing individual-level consistency or ensuring equitable utility across groups defined by protected attributes~\cite{userfairness1,userfairness2,userfairness3}. In contrast, \textit{provider-side fairness} addresses systematic exposure imbalances, where a small subset of providers receives the majority of recommendation traffic while others remain persistently under-exposed~\cite{new6,new5,new8,new9}. 

Existing approaches to provider exposure fairness can be broadly categorized into \textit{in-processing} and \textit{post-processing} approaches, depending on how fairness constraints are integrated into the recommendation pipeline. We review these two lines of work below.

\textit{In-processing approaches} integrate fairness constraints directly into model training through techniques such as reweighting~\cite{FairDual,Ada2Fair,cai1} or fairness-aware regularization~\cite{Multi-FR}. For example, Ada2Fair~\cite{Ada2Fair} employs an adaptive weight generator to adjust the optimization scale of interaction samples, boosting the exposure of niche providers during training. FairDual~\cite{FairDual} applies dual-mirror gradient descent to dynamically compute sample weights, explicitly targeting support for worst-off groups. Multi-FR~\cite{Multi-FR} incorporates fairness regularizers into the learning objective and employs multi-gradient descent to improve exposure for unpopular items. While effective, these approaches tightly couple fairness objectives with model parameters, making full retraining inevitable when fairness requirements change. 

\textit{Post-processing approaches} operate on the outputs of pretrained recommenders and adjust recommendation lists to improve provider exposure fairness. Typical techniques include reranking strategies and exposure allocation mechanisms. For example, FairRec~\cite{FairRec} formulates exposure allocation using equilibrium-based criteria to guarantee minimum exposure for each provider, and CPFair~\cite{CPFair} casts fair reranking as a knapsack optimization problem. P-MMF~\cite{p-mmf} models exposure as a resource allocation process optimized via momentum-based gradient descent for online recommendation scenarios, and LTP-MMF~\cite{LTP-MMF} extends this framework to long-term fairness under feedback loops. In addition, some post-processing methods target related but distinct notions of fairness at the item level. TaxRank~\cite{TaxRank} adopts a taxation-inspired mechanism optimized via optimal transport, while ElasticRank~\cite{ElasticRank} introduces elasticity to characterize the accuracy–fairness trade-off. Although post-processing methods are model-agnostic and avoid retraining, they typically embed fixed fairness criteria into the reranking process, limiting their adaptability to diverse fairness requirements.

In contrast, we treat provider exposure fairness as a post-hoc adaptation problem. We propose PFA, a lightweight framework that freezes the pretrained recommender and attaches a learnable fairness adapter to produce additive score corrections that steer recommendations toward fairer provider exposure distributions. Different from in-processing methods, PFA decouples fairness optimization from model training, avoiding costly retraining. Unlike conventional post-processing methods, PFA leverages a learnable fairness adapter together with HEFA, a flexible fairness objective, enabling dynamic adaptation to diverse fairness requirements through learning rather than fixed reranking criteria.

\begin{figure*}[!h]
  \centering  
  \includegraphics[width=\textwidth]{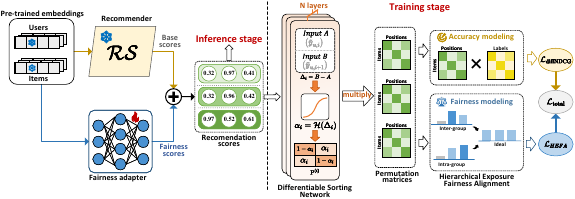}
  \caption{Overview of the proposed PFA framework. The frozen base recommender provides base scores, which are adjusted by a lightweight fairness adapter to produce fairness-aware scores. The differentiable sorting network enables end-to-end optimization of both the HEFA loss and the differentiable NDCG loss.}
  \label{fig:framework}
\end{figure*}

\section{Preliminaries}
\label{sec:preliminaries}
\subsection{Formulation}
\label{sec:formulation}
We consider a general recommender system scenario. Let $\mathcal{U}$, $\mathcal{V}$, and $\mathcal{S}$ denote the sets of users, items, and providers, respectively, with $|\mathcal{U}| = M$, $|\mathcal{V}| = N$, and $|\mathcal{S}| = L$. Historical user–item interactions are represented by a binary matrix $\mathbf{R} \in \{0,1\}^{M \times N}$, where $R_{u,v} = 1$ indicates that user $u$ has interacted with item $v$, and $R_{u,v} = 0$ otherwise. Each item $v$ is associated with a provider $s_v$ (e.g., the seller on an e-commerce platform), and the set of items supplied by provider $s$ is denoted as $\mathcal{V}_s = \{ v \in \mathcal{V} \mid s_v = s \}$. The recommender learns a scoring function $f_\theta$ that predicts user $u$'s preference for item $v$ based on their embeddings $\mathbf{e}_u, \mathbf{e}_v \in \mathbb{R}^d$, yielding a predicted score $\hat{y}_{u,v} = f_\theta(\mathbf{e}_u, \mathbf{e}_v)$. For each user $u$, the system generates a personalized recommendation list $L_u$ by selecting the top-$K$ items ranked according to $\hat{y}_{u,v}$.

\textbf{Provider Exposure.} We now formalize provider exposure. Intuitively, a provider gains exposure when its items appear in recommendation lists. Let $\pi(v \mid L_u)$ denote the exposure assigned to item $v$ in user $u$'s list $L_u$. The total exposure of provider $s$ is defined as the sum of exposure scores received by all its items across all users~\cite{attention1,attention3,Ada2Fair}:
\begin{equation}
\label{eq:exposure}
e_s = \sum_{u \in \mathcal{U}} \sum_{v \in \mathcal{V}_s} \pi(v \mid L_u).
\end{equation}
Since higher-ranked items attract more user attention, we model position bias using a logarithmic decay function following prior works~\cite{attention1,attention2,attention3}. Let $\mathrm{rank}(v, L_u) \in \{1, \dots, K\}$ denote the rank position of item $v$ in $L_u$. The exposure score is defined as:
\begin{equation}
\label{eq:position_bias}
\pi(v \mid L_u) = 
\begin{cases} 
\dfrac{1}{\log_2(1 + \mathrm{rank}(v, L_u))}, & \text{if } v \in L_u, \\ 
0, & \text{otherwise}. 
\end{cases}
\end{equation}

\textbf{Fairness Target.} Our fairness objective aims to regulate provider exposure according to a predefined target distribution. As a default instantiation, following prior work~\cite{attention3,Multi-FR,Target1}, we adopt a uniform target distribution:
\begin{equation}
\label{eq:target}
\mathbf{t} = [t_1, \dots, t_L], \quad \text{where} \quad t_s = \frac{1}{L}, \; \forall s \in \mathcal{S}.
\end{equation}
This target embodies the principle of equal opportunity, ensuring that each provider has an equal chance of receiving exposure. We note that our framework naturally supports alternative targets~\cite{attention3,t1,t2} (e.g., proportional to catalog size) by simply specifying a different distribution $\mathbf{t}$. The choice of target depends on the specific application context and platform objectives.

\subsection{Pretrained Recommender}
\label{sec:pretrained}
Our framework assumes a pretrained recommender that produces user and item embeddings, such as matrix factorization models~\cite{MF} and graph-based recommenders~\cite{LightGCN,SGL}. The backbone model is pretrained using the Bayesian Personalized Ranking (BPR) loss~\cite{BPR}, which encourages observed interactions to be ranked higher than unobserved ones. Let $\mathcal{D} = \{(u, v^+, v^-)\}$ denote the training set, where $v^+$ is an interacted item and $v^-$ is a sampled negative item for user $u$. The BPR loss is defined as:
\begin{equation}
\mathcal{L}_{\text{BPR}} = -\sum_{(u, v^+, v^-) \in \mathcal{D}} 
\ln \sigma\bigl(\hat{y}_{u,v^+} - \hat{y}_{u,v^-}\bigr),
\end{equation}
where $\sigma(\cdot)$ is the sigmoid function. After pretraining, we obtain user and item embeddings $\mathbf{e}_u, \mathbf{e}_v \in \mathbb{R}^d$. These embeddings are frozen during the subsequent fairness adapter training.

\subsection{Differentiable Sorting Network}
\label{sec:dsn}
Differentiable sorting networks~\cite{DSN,MDSN} provide a smooth relaxation of discrete sorting operations, enabling gradient-based optimization. Unlike alternative differentiable sorting methods~\cite{NeuralSort, FDS} that produce unimodal row-stochastic matrices, differentiable sorting networks yield \textbf{doubly stochastic permutation matrices} that faithfully preserve probability distributions. In our framework, we leverage the differentiable sorting networks to compute provider exposure and optimize ranking metrics such as NDCG, both of which depend on item positions in recommendation lists.

A differentiable sorting network approximates sorting through multiple layers of compare-and-swap operations. For a user $u$, let $\tilde{\mathbf{y}}_u \in \mathbb{R}^N$ denote the predicted scores for all items. The network iteratively compares adjacent element pairs and conditionally swaps them to progressively approximate a sorted sequence. Taking the odd-even transposition network~\cite{odd} as an example, odd layers compare pairs at positions $(1,2), (3,4), \dots$, while even layers compare $(2,3), (4,5), \dots$. For each comparison, a swap operation is performed to enforce descending order. However, standard swap operations rely on non-differentiable $\max$ and $\min$ functions, which we relax using a soft formulation.

For an adjacent pair $(\tilde{y}_{u,i}, \tilde{y}_{u,i+1})$, a hard swap is defined as: $\tilde{y}'_{u,i} = \max(\tilde{y}_{u,i}, \tilde{y}_{u,i+1})$ and $\tilde{y}'_{u,i+1} = \min(\tilde{y}_{u,i}, \tilde{y}_{u,i+1})$. To enable differentiation, we replace $\max$/$\min$ with soft interpolation $\alpha_i \in [0, 1]$:
\begin{equation}
\begin{aligned}
\tilde{y}'_{u,i} &= (1 - \alpha_i) \cdot \tilde{y}_{u,i} + \alpha_i \cdot \tilde{y}_{u,i+1}, \\
\tilde{y}'_{u,i+1} &= \alpha_i \cdot \tilde{y}_{u,i} + (1 - \alpha_i) \cdot \tilde{y}_{u,i+1}.
\end{aligned}
\end{equation}
The interpolation weight is computed as $\alpha_i = \mathcal{H}(\tilde{y}_{u,i+1} - \tilde{y}_{u,i})$, where $\mathcal{H}(x) = \frac{1}{\pi} \arctan(\beta \cdot x) + \frac{1}{2}$ is a Cauchy-based smoothing function~\cite{MDSN,PreferenceAlignment} and $\beta > 0$ controls sharpness. Intuitively, $\alpha_i$ approaches 0 when the pair is already correctly ordered, and approaches 1 when a swap is needed. Larger $\beta$ yields behavior closer to hard sorting, while smaller $\beta$ provides smoother gradients.

Each soft swap can be represented as a $2 \times 2$ doubly stochastic matrix $\mathbf{P}^{(i)} = \bigl[\begin{smallmatrix} 1-\alpha_i & \alpha_i \\ \alpha_i & 1-\alpha_i \end{smallmatrix}\bigr]$. For each differentiable sorting network layer, these local matrices are embedded into an $N \times N$ identity matrix to form a layer-wise permutation matrix $\mathbf{P}_k$. The overall soft permutation matrix is computed by sequentially composing all $K$ layers: $\mathbf{P}_{\mathrm{soft}} = \mathbf{P}_1 \mathbf{P}_2 \cdots \mathbf{P}_K \in \mathbb{R}^{N \times N}$. Each element $[\mathbf{P}_{\mathrm{soft}}]_{jk}$ represents the probability that item $j$ is assigned to rank position $k$, enabling differentiable computation of exposure-based fairness objectives and ranking metrics. We refer readers to~\cite{DSN} for implementation details.

\section{Methodology}
\subsection{Overview}
\label{sec:overview}
As illustrated in Figure~\ref{fig:framework}, the proposed PFA framework comprises two components: a \textit{frozen pretrained recommender} and a \textit{lightweight fairness adapter}. The pretrained recommender produces original user-item preference scores, while the fairness adapter generates additive score corrections to steer recommendations toward fair exposure distributions. During training, we optimize the adapter using the \textit{HEFA objective}, which decomposes fairness optimization into inter-group and intra-group levels, along with a differentiable NDCG loss to preserve recommendation accuracy. A \textit{differentiable sorting network} enables end-to-end optimization of both fairness and ranking objectives. During inference, the final scores, obtained by combining the original recommender outputs with the adapter's corrections, are directly used for recommendation without additional overhead. We elaborate on each component in the following subsections.

\subsection{Fairness Adapter}
\label{sec:adapter}
The fairness adapter is a lightweight trainable module that operates on frozen user and item embeddings to learn additive score corrections for fairness-aware recommendation. As illustrated in Figure~\ref{fig:framework}, the adapter takes embeddings produced by the pretrained backbone as input. For a batch of $B$ users, let $\mathbf{E}_u \in \mathbb{R}^{B \times d}$ denote the user embeddings and $\mathbf{E}_v \in \mathbb{R}^{B \times N \times d}$ denote the embeddings of $N$ candidate items, where $d$ is the embedding dimension. To enable feature interaction, the user embeddings are broadcast along the item dimension to obtain $\tilde{\mathbf{E}}_u \in \mathbb{R}^{B \times N \times d}$, and concatenated with item embeddings:
\begin{equation}
\mathbf{Z} = [\tilde{\mathbf{E}}_u \,\|\, \mathbf{E}_v] \in 
\mathbb{R}^{B \times N \times 2d}.
\end{equation}
The concatenated representation is processed by a three-layer MLP with hidden dimension $h$:
\begin{equation}
\begin{aligned}
\mathbf{H}_1 &= \text{ReLU}(\mathbf{Z} \mathbf{W}_1 + \mathbf{b}_1), \\
\mathbf{H}_2 &= \text{ReLU}(\mathbf{H}_1 \mathbf{W}_2 + \mathbf{b}_2), \\
\mathbf{\Delta} &= \mathbf{H}_2 \mathbf{W}_3 + \mathbf{b}_3,
\end{aligned}
\end{equation}
where $\mathbf{W}_1 \in \mathbb{R}^{2d \times h}$, $\mathbf{W}_2 \in \mathbb{R}^{h \times h}$, $\mathbf{W}_3 \in \mathbb{R}^{h \times 1}$, and $\mathbf{b}_i$ are the corresponding biases. The output $\mathbf{\Delta} \in \mathbb{R}^{B \times N}$ represents the score corrections. The final adjusted scores are obtained by adding these corrections to the base scores $\hat{\mathbf{Y}} \in \mathbb{R}^{B \times N}$ from the pretrained recommender:
\begin{equation}
\tilde{\mathbf{Y}} = \hat{\mathbf{Y}} + \mathbf{\Delta}.
\end{equation}
Equivalently, for each user $u$, we have $\tilde{\mathbf{y}}_u = \hat{\mathbf{y}}_u + \Delta\mathbf{y}_u$, where $\Delta\mathbf{y}_u \in \mathbb{R}^N$ is the corresponding row of $\mathbf{\Delta}$. During inference, the adjusted scores $\tilde{\mathbf{Y}}$ are directly used for ranking without additional overhead. During training, we apply the differentiable sorting network (Section~\ref{sec:dsn}) to the final adjusted scores to obtain soft permutation matrices, enabling end-to-end optimization of both exposure fairness and ranking accuracy.

\subsection{Hierarchical Exposure Fairness Alignment}
\label{sec:hefa}
With the fairness adapter in place, we now design the training objective to optimize provider exposure fairness. A natural approach is to minimize the KL divergence between the provider exposure distribution and a target fair distribution. However, this global alignment implicitly treats all providers equally, overlooking structural disparities such as imbalanced group sizes and heterogeneous exposure within groups. Consequently, fairness may appear satisfied at an aggregate level while substantial inter- and intra-group imbalances persist. To address this limitation, we propose HEFA, which decomposes fairness optimization into inter-group and intra-group components, enabling flexible control at different granularities via tunable weights.

\textbf{Differentiable Exposure Estimation.} We first compute provider exposure scores in a differentiable manner using the soft permutation matrix. Recall from Section~\ref{sec:dsn} that the soft permutation matrix $\mathbf{P}_{\mathrm{soft}}^{(u)} \in \mathbb{R}^{N \times N}$ encodes the probability that item $v$ is ranked at position $k$ via $[\mathbf{P}_{\mathrm{soft}}^{(u)}]_{vk}$. The expected exposure of provider $s$ is defined as: 
\begin{equation}
\hat{e}_s = \sum_{u \in \mathcal{U}} \sum_{v \in \mathcal{V}_s} \sum_{k=1}^{K} [\mathbf{P}_{\mathrm{soft}}^{(u)}]_{vk} \cdot b_k,
\end{equation}
where $b_k = 1/\log_2(1+k)$ is the position bias at rank $k$. This is the differentiable counterpart of $e_s$ in Eq.~\eqref{eq:exposure}: the inner sum computes the expected position-weighted exposure for item $v$, which is then aggregated over all items of provider $s$ and all users. Normalizing yields the provider exposure distribution $\mathbf{p} = [p_s]_{s \in \mathcal{S}}$ with $p_s = \hat{e}_s / \sum_{s'} \hat{e}_{s'}$. 

Let $\mathbf{t} = [t_s]_{s \in \mathcal{S}}$ denote the target fair exposure distributions defined in Section~\ref{sec:formulation}. A natural objective is to minimize the global KL divergence, which measures how the actual exposure distribution deviates from the target:
\begin{equation}
\label{eq:kl_global}
\mathcal{L}_{\mathrm{KL}} = D_{\mathrm{KL}}(\mathbf{p} \| \mathbf{t}) = \sum_{s \in \mathcal{S}} p_s \log \frac{p_s}{t_s}.
\end{equation}

\textbf{Hierarchical Decomposition.} We partition the provider set $\mathcal{S}$ into $C$ disjoint groups $\mathcal{G} = \{G_1, \dots, G_C\}$ based on historical exposure levels (e.g., high-exposure vs. low-exposure providers). For each group $G_c$, let $p^G_c = \sum_{s \in G_c} p_s$ denote its total exposure share. We then define:

\textbf{Inter-group level:}
\begin{itemize}
    \item Achieved distribution: $\mathbf{p}^G = [p^G_1, \dots, p^G_C]$
    \item Target distribution: $\mathbf{t}^G = [t^G_1, \dots, t^G_C]$, the desired group-level exposure shares
    \item Aggregated target: $\bar{t}^G_c = \sum_{s \in G_c} t_s$, the sum of provider-level targets in $G_c$
\end{itemize}

\textbf{Intra-group level:}
\begin{itemize}
    \item Achieved distribution: $\mathbf{p}^{(c)} = [p_s / p^G_c]_{s \in G_c}$, the relative exposure of providers within group $G_c$.
    \item Target distribution: $\mathbf{t}^{(c)} = [t_s / \bar{t}^G_c]_{s \in G_c}$, the relative target shares within group $G_c$.
\end{itemize}

\begin{theorem}[Hierarchical KL Decomposition]
\label{thm:kl_decomp}
The global KL divergence between the provider exposure distribution $\mathbf{p}$ and target distribution $\mathbf{t}$ can be decomposed as:
\begin{equation}
\label{eq:kl_decomp}
D_{\mathrm{KL}}(\mathbf{p} \| \mathbf{t}) = \underbrace{D_{\mathrm{KL}}(\mathbf{p}^G \| \mathbf{t}^G)}_{\textit{inter-group}} + \underbrace{\sum_{c=1}^{C} p^G_c \cdot D_{\mathrm{KL}}(\mathbf{p}^{(c)} \| \mathbf{t}^{(c)})}_{\textit{intra-group}} + \underbrace{\sum_{c=1}^{C} p^G_c \log \frac{t^G_c}{\bar{t}^G_c}}_{\Delta_{\mathrm{calib}}},
\end{equation}
where the inter-group term measures divergence across provider groups, the intra-group term captures within-group disparities, and $\Delta_{\mathrm{calib}}$ is a calibration term accounting for the mismatch between group-level and aggregated provider targets.
\end{theorem}

\begin{proof}
For each provider $s \in G_c$, we express the exposure and target in terms of inter- and intra-group components. Specifically, since $p^{(c)}_s = p_s / p^G_c$ by definition, we have $p_s = p^G_c \cdot p^{(c)}_s$; similarly, since $t^{(c)}_s = t_s / \bar{t}^G_c$, we obtain $t_s = \bar{t}^G_c \cdot t^{(c)}_s$. Substituting into the log-ratio:
\begin{equation}
\log \frac{p_s}{t_s} = \log \frac{p^G_c \cdot p^{(c)}_s}{\bar{t}^G_c \cdot t^{(c)}_s} = \log \frac{p^G_c}{\bar{t}^G_c} + \log \frac{p^{(c)}_s}{t^{(c)}_s}.
\end{equation}
Substituting into Eq.~\eqref{eq:kl_global}, regrouping by provider groups, and using the fact that $\sum_{s \in G_c} p^{(c)}_s = 1$ yields:
\begin{align}
D_{\mathrm{KL}}(\mathbf{p} \| \mathbf{t}) &= \sum_{c=1}^{C} \sum_{s \in G_c} p^G_c \cdot p^{(c)}_s \left( \log \frac{p^G_c}{\bar{t}^G_c} + \log \frac{p^{(c)}_s}{t^{(c)}_s} \right) \notag \\
&= \sum_{c=1}^{C} p^G_c \log \frac{p^G_c}{\bar{t}^G_c} + \sum_{c=1}^{C} p^G_c \cdot D_{\mathrm{KL}}(\mathbf{p}^{(c)} \| \mathbf{t}^{(c)}).
\end{align}
Introducing the group-level target $\mathbf{t}^G$ into the first term:
\begin{equation}
\sum_{c=1}^{C} p^G_c \log \frac{p^G_c}{\bar{t}^G_c} = \sum_{c=1}^{C} p^G_c \log \frac{p^G_c}{t^G_c} + \sum_{c=1}^{C} p^G_c \log \frac{t^G_c}{\bar{t}^G_c} = D_{\mathrm{KL}}(\mathbf{p}^G \| \mathbf{t}^G) + \Delta_{\mathrm{calib}}.
\end{equation}
Combining all terms yields Eq.~\eqref{eq:kl_decomp}.
\end{proof}

The calibration term $\Delta_{\mathrm{calib}} = \sum_{c=1}^{C} p^G_c \log \frac{t^G_c}{\bar{t}^G_c}$ accounts for the mismatch between the group-level target $t^G_c$ and the aggregated provider-level target $\bar{t}^G_c$. When $t^G_c = \bar{t}^G_c$ for all groups, this term vanishes, and minimizing the sum of the inter-group and intra-group terms exactly recovers the global KL divergence objective. When $t^G_c \neq \bar{t}^G_c$, the decomposition separates unavoidable target misalignment from controllable exposure disparities, yielding a principled relaxation. This allows platforms to enforce group-level fairness objectives while explicitly accounting for cross-level calibration mismatch.

Based on this decomposition, we define the Hierarchical Exposure Fairness Alignment (HEFA) loss:
\begin{equation}
\label{eq:hefa_loss}
\mathcal{L}_{\mathrm{HEFA}} = \lambda_{\mathrm{inter}} \cdot D_{\mathrm{KL}}(\mathbf{p}^G \| \mathbf{t}^G) + \lambda_{\mathrm{intra}} \cdot \sum_{c=1}^{C} p^G_c \cdot D_{\mathrm{KL}}(\mathbf{p}^{(c)} \| \mathbf{t}^{(c)}),
\end{equation}
where $\lambda_{\mathrm{inter}}, \lambda_{\mathrm{intra}} > 0$ control the relative importance of inter-group and intra-group fairness. Setting $\lambda_{\mathrm{inter}} > \lambda_{\mathrm{intra}}$ prioritizes inter-group fairness, while the reverse emphasizes intra-group equity.

\textbf{Advantage.} HEFA offers two-dimensional flexibility absent in direct KL minimization. First, the target distribution $\mathbf{t}^G$ specifies desired inter-group exposure allocation. For instance, setting $t^G_c = 1/C$ enforces \emph{group parity} that compensates historically disadvantaged groups, while $t^G_c \propto |G_c|$ yields size-proportional allocation. Second, the weights $\lambda_{\text{inter}}$ and $\lambda_{\text{intra}}$ enable fine-grained trade-offs between inter-group balance and intra-group uniformity. This design allows HEFA to instantiate diverse fairness policies by specifying different targets and weights, without modifying the optimization framework. HEFA reduces to global KL minimization only when $t^G_c = \bar{t}^G_c$ for all $c$ and $\lambda_{\text{inter}} = \lambda_{\text{intra}}$; otherwise, it provides hierarchical control over fairness that is unattainable with direct KL optimization.

\subsection{Differentiable NDCG Loss}
\label{sec:ndcg}
Since fairness-oriented score corrections may degrade ranking accuracy, we jointly optimize the HEFA loss with a differentiable NDCG loss to preserve recommendation quality. NDCG effectively captures both item relevance and position sensitivity~\cite{diffNDCG1,diffNDCG2}, making it a natural choice for the accuracy objective.

Let $\mathbf{r}_u = [r_{u,1}, \dots, r_{u,N}]^\top$ denote the ground-truth relevance labels for user $u$, where $r_{u,v}$ indicates the relevance of item $v$ to user $u$ derived from user interactions. The standard NDCG@$K$ is defined as:
\begin{equation}
\label{eq:ndcg}
\mathrm{NDCG}@K = \frac{\mathrm{DCG}@K}{\mathrm{IDCG}@K} = 
\frac{1}{\mathrm{IDCG}@K} \sum_{k=1}^{K} \frac{2^{r_{\pi(k)}} - 1}{\log_2(1 + k)},
\end{equation}
where $\pi(k)$ denotes the index of the item ranked at position $k$, and IDCG 
is the ideal DCG computed by ranking items in descending order of relevance:
\begin{equation}
\label{eq:idcg}
\mathrm{IDCG}@K = \sum_{k=1}^{K} \frac{2^{r_{\pi^*(k)}} - 1}{\log_2(1 + k)},
\end{equation}
where $\pi^*(k)$ denotes the index of the item at position $k$ in the ideal ranking. Since IDCG depends solely on the ground-truth labels, it remains constant during training. 

However, NDCG is non-differentiable due to the discrete sorting operation required for computing $\pi(k)$. To address this, we first reformulate NDCG using a permutation matrix. Let $\mathbf{P}_{\mathrm{hard}}^{(u)} \in \{0,1\}^{N \times N}$ denote the hard permutation matrix, where $[\mathbf{P}_{\mathrm{hard}}^{(u)}]_{vk} = 1$ if item $v$ is placed at position $k$, so that the relevance at position $k$ can be written as $r_{\pi(k)} = [(\mathbf{P}_{\mathrm{hard}}^{(u)})^\top \mathbf{r}_u]_k$. As introduced in Section~\ref{sec:dsn}, the differentiable sorting network produces a soft permutation matrix $\mathbf{P}_{\mathrm{soft}}^{(u)}$, where each element $[\mathbf{P}_{\mathrm{soft}}^{(u)}]_{vk}$ represents the probability that item $v$ is assigned to position $k$. Replacing the hard permutation with $\mathbf{P}_{\mathrm{soft}}^{(u)}$, obtained by applying the differentiable sorting network to the adjusted scores $\tilde{\mathbf{y}}_u$, yields the relaxed relevance:
\begin{equation}
\hat{r}_k = [(\mathbf{P}_{\mathrm{soft}}^{(u)})^\top \mathbf{r}_u]_k.
\end{equation}
Substituting into the NDCG formula gives the differentiable version:
\begin{equation}
\mathrm{diffNDCG}(\tilde{\mathbf{y}}_u, \mathbf{r}_u) = 
\frac{1}{\mathrm{IDCG}@K} \sum_{k=1}^{K} 
\frac{2^{\hat{r}_k} - 1}{\log_2(1 + k)}.
\end{equation}
Since higher NDCG indicates better ranking quality, we define the differentiable NDCG loss over all users as:
\begin{equation}
\mathcal{L}_{\mathrm{diffNDCG}} = \frac{1}{|\mathcal{U}|} \sum_{u \in \mathcal{U}} 
\left(1 - \mathrm{diffNDCG}(\tilde{\mathbf{y}}_u, \mathbf{r}_u)\right).
\end{equation}

\subsection{Overall Objective}
\label{sec:objective}
To jointly optimize provider exposure fairness and ranking accuracy, we combine the HEFA loss and the diffNDCG loss into a unified objective:
\begin{equation}
\mathcal{L}_{\mathrm{total}} = \mathcal{L}_{\mathrm{HEFA}} + \lambda_{\mathrm{acc}} 
\cdot \mathcal{L}_{\mathrm{diffNDCG}},
\end{equation}
where $\lambda_{\mathrm{acc}} > 0$ controls the trade-off between fairness and accuracy. A larger value emphasizes ranking quality, while a smaller value prioritizes fairness optimization.

\section{Experiments}
\subsection{Experimental Setup}
\label{Setup}
\subsubsection{Datasets.}
We evaluate our method on three widely-used real-world datasets: Beauty and Movies \& TV from the Amazon Reviews corpus~\footnote{\label{fn:amazon}\url{https://jmcauley.ucsd.edu/data/amazon/}}, and RateBeer~\footnote{\url{https://snap.stanford.edu/data/web-RateBeer.html}}. To assess provider fairness, we identify providers based on item attributes, treating brands as providers in the Amazon datasets and brewers as providers in RateBeer. Following standard practice~\cite{fairdiverse,providerfairness1,wei}, we filter out users and items with fewer than five interactions. For each user, we randomly split interactions into 70\% for training, 10\% for validation, and 20\% for testing. Detailed dataset statistics are reported in Table~\ref{Statistical}.
\begin{table}[htbp]
    \centering
    \caption{Statistics of the datasets used in our experiments.}
    \renewcommand{\arraystretch}{0.9}
    \setlength\tabcolsep{1.0pt}
    \begin{tabular}{c|ccccc}
    \hline
                  & Users & Items & Providers & Interactions & Sparsity \\ \hline
    Beauty        & 11,742 & 6,506  & 196       & 98,404        & 99.87\%  \\
    Movies \& TV & 10,111 & 3,795  & 77        & 139,568       & 99.64\%  \\
    RateBeer      & 12,483 & 48,971 & 1,645      & 2,686,981      & 99.56\%  \\ \hline
    \end{tabular}
    \label{Statistical}
\end{table}

\subsubsection{Evaluation Metrics.}
Following previous work~\cite{Ada2Fair,fairdiverse,Multi-FR}, we evaluate PFA from two aspects: recommendation accuracy and provider exposure fairness. For accuracy, we report NDCG, HR and MRR, where higher values indicate better ranking performance. For fairness, we adopt Gini, Ent., and CV to measure how evenly exposure is distributed across providers:
\begin{itemize}
\item \textbf{Gini Index (Gini):} Gini measures exposure inequality among providers. Lower values indicate a more equal distribution, while higher values imply stronger concentration on a small number of providers~\cite{fairdiverse,Multi-FR}.
\item \textbf{Entropy (Ent.):} Ent. captures the overall uniformity of provider exposure. Higher entropy indicates more evenly distributed exposure and thus better fairness~\cite{fairdiverse,Survey3}.
\item \textbf{Coefficient of Variation (CV):} CV reflects the relative variation of exposure across providers. It is non-negative and equals zero when all providers receive identical exposure. A smaller CV thus indicates more balanced exposure and improved fairness~\cite{Ada2Fair,Survey3}.
\end{itemize}

\begin{table*}[!t]
\centering
\caption{
Performance comparison of PFA and fairness-aware baselines on three datasets using BPRMF as the base model. FairRec and FairDual are marked as N/A on RateBeer due to computational timeout. The best results among fairness-aware methods are highlighted in \textbf{bold}, and the second-best results are \underline{underlined}.}
\scriptsize
\setlength\tabcolsep{2.5pt}
\renewcommand{\arraystretch}{1.0}
\begin{tabular}{cl|cccccc|cccccc|cccccc}
\toprule
\centering
\multirow{3}{*}{\textbf{Type}} &
\multirow{3}{*}{\textbf{Method}} &
\multicolumn{6}{c|}{\textbf{Amazon Beauty}} &
\multicolumn{6}{c|}{\textbf{Amazon Movies \& TV}} &
\multicolumn{6}{c}{\textbf{RateBeer}} \\
\cmidrule(lr){3-8} \cmidrule(lr){9-14} \cmidrule(lr){15-20}
& &
\multicolumn{3}{c}{\textbf{Rec. Acc.}$\uparrow$} &
\multicolumn{3}{c|}{\textbf{Prov. Fair.}} &
\multicolumn{3}{c}{\textbf{Rec. Acc.}$\uparrow$} &
\multicolumn{3}{c|}{\textbf{Prov. Fair.}} &
\multicolumn{3}{c}{\textbf{Rec. Acc.}$\uparrow$} &
\multicolumn{3}{c}{\textbf{Prov. Fair.}} \\
\cmidrule(lr){3-5} \cmidrule(lr){6-8}
\cmidrule(lr){9-11} \cmidrule(lr){12-14}
\cmidrule(lr){15-17} \cmidrule(lr){18-20}
& & NDCG & HR & MRR & Gini$\downarrow$ & Ent.$\uparrow$ & CV$\downarrow$
& NDCG & HR & MRR & Gini$\downarrow$ & Ent.$\uparrow$ & CV$\downarrow$
& NDCG & HR & MRR & Gini$\downarrow$ & Ent.$\uparrow$ & CV$\downarrow$ \\
\midrule
-- & Base model          
& \textbf{0.1000} & \textbf{0.1870} & \textbf{0.0833} & 0.7260 & 6.0759 & 2.0746 
& \textbf{0.1137} & \textbf{0.2500} & \textbf{0.0931} & 0.7727 & 4.5499 & 1.9664
& \textbf{0.4952} & \textbf{0.6122} & \textbf{0.5713} & 0.9424 & 6.9071 & 4.8245 \\
\midrule
\multirow{4}{*}{\textit{In-processing}}
& SDRO [WWW'22]
& 0.0935 & 0.1747 & 0.0779 & \underline{0.6305} & 6.3448 & 2.1748 
& 0.0995 & 0.2067 & 0.0851 & 0.6598 & \underline{5.0969} & 1.5818
& 0.4152 & 0.4778 & 0.4921 & \underline{0.8625} & 7.6257 & 4.7581 \\
& Multi-FR [TOIS'22]    
& 0.0912 & 0.1714 & 0.0773 & 0.6641 & 6.2726 & 2.1927 
& 0.0947 & 0.1606 & 0.0852 & 0.6724 & 4.7715 & 1.8527
& 0.3822 & 0.4358 & 0.4602 & 0.8737 & 7.5815 & 4.7920 \\
& Ada2Fair [RecSys'24] 
& 0.0872 & 0.1627 & 0.0758 & 0.6426 & 6.3221 & 2.2855 
& 0.0902 & 0.1505 & 0.0817 & 0.6526 & 4.8428 & 1.7511
& 0.3751 & 0.4511 & 0.4819 & 0.8805 & 7.4628 & 4.6263 \\
& FairDual [ICLR'25] 
& 0.0934 & 0.1692 & \underline{0.0801} & 0.6457 & 6.3202 & 2.3419 
& 0.0969 & 0.1815 & \underline{0.0908} & 0.6519 & 4.9753 & 1.8198
& \multicolumn{6}{c}{N/A} \\
\midrule
\multirow{5}{*}{\textit{Post-processing}}
& FairRec [WWW'20]    
& 0.0852 & 0.1687 & 0.0611 & 0.6861 & 6.2124 & 2.0970 
& 0.1057 & 0.2132 & 0.0856 & 0.6885 & 4.9364 & 1.7917
& \multicolumn{6}{c}{N/A} \\
& CPFair [SIGIR'22]                     
& 0.0932 & 0.1745 & 0.0731 & 0.6426 & \underline{6.3468} & \textbf{1.5632} 
& 0.0944 & 0.2144 & 0.0758 & \underline{0.6469} & 5.0542 & \underline{1.5100}
& 0.3909 & 0.4856 & 0.4735 & 0.8732 & \textbf{8.0953} & \underline{2.9991} \\
& TaxRank [SIGIR'24]                 
& 0.0855 & 0.1636 & 0.0675 & 0.6542 & 6.3173 & \underline{1.7859}
& 0.0965 & 0.2122 & 0.0729 & 0.6518 & 5.0641 & \textbf{1.4476}
& 0.4087 & 0.4903 & 0.4693 & 0.8639 & 7.6918 & 4.2295 \\
& ElasticRank [SIGIR'25]            
& 0.0885 & 0.1715 & 0.0690 & 0.6420 & 5.7372 & 3.3297 
& 0.0929 & 0.1972 & 0.0832 & 0.6931 & 3.7128 & 4.0306
& 0.3624 & 0.4151 & 0.4901 & 0.8775 & 5.6857 & \textbf{1.0721} \\
& PFA [Ours]      
& \underline{0.0943} & \underline{0.1751} & 0.0780 & \textbf{0.6282} & \textbf{6.3544} & 2.2524 
& \underline{0.1066} & \underline{0.2172} & 0.0863 & \textbf{0.6403} & \textbf{5.1010} & 1.7457
& \underline{0.4284} & \underline{0.4964} & \underline{0.5161} & \textbf{0.8467} & \underline{7.7781} & 4.7565 \\
\bottomrule
\end{tabular}
\label{tab:main_results}
\end{table*}

\subsubsection{Base Recommenders and Baselines.}
We apply PFA to three widely used models, BPRMF~\cite{BPR}, LightGCN~\cite{LightGCN}, and SGL~\cite{SGL}, to test its generalizability. To evaluate its effectiveness in both recommendation accuracy and provider fairness, we compare PFA with eight representative fairness-aware baselines, grouped into two categories: (1) \emph{in-processing} methods: SDRO~\cite{SDRO}, Multi-FR~\cite{Multi-FR}, Ada2Fair~\cite{Ada2Fair}, and FairDual~\cite{FairDual}; and (2) \emph{post-processing} methods: FairRec~\cite{FairRec}, CPFair~\cite{CPFair}, TaxRank~\cite{TaxRank}, and ElasticRank~\cite{ElasticRank}.

\subsubsection{Implementation Details.} 
For all backbone recommenders, we tune the learning rate within $\{10^{-1}, 10^{-2}, 10^{-3}, 10^{-4}, 10^{-5}\}$. For GNN-based models, we additionally search the number of graph convolutional layers in $\{1, 2, 3\}$. All fairness-aware baselines are implemented based on the original code released by the authors, with reference to the FairDiverse toolkit~\cite{fairdiverse}. Hyperparameters are initialized from the original papers and re-tuned on the validation set. We use the Adam optimizer~\cite{adamo} with a mini-batch size of 256 and set the embedding dimension to 32 for all methods unless otherwise specified. For PFA, we set $\lambda_{\mathrm{inter}} = \lambda_{\mathrm{intra}} = 1$ and tune $\lambda_{\mathrm{acc}}$ in $\{10^{-6}, \dots, 10^{-1}\}$. The learning rate of the fairness adapter is tuned within $\{10^{-5}, 10^{-4}, 10^{-3}\}$. For the differentiable sorting network (Section~\ref{sec:dsn}), we set the steepness parameter $\beta = 10.0$~\cite{MDSN,PreferenceAlignment}. Following prior work on provider group analysis~\cite{longterm,Multi-FR}, we rank providers by their training-set exposure counts and partition them into head (top 20\%), mid (middle 60\%), and tail (bottom 20\%) groups. For the group-level target distribution in HEFA, we set $t^G_c = 1/C$ for all groups, enforcing equal exposure allocation across the three provider groups. This serves as a representative policy for fair comparison. Alternative fairness requirements can be supported by specifying different targets without modifying the framework. During evaluation, the recommendation list length is set to $K=20$, and all results are averaged over three independent runs to ensure robustness. We select the checkpoint that achieves the best accuracy-fairness trade-off on the validation set and report results using the corresponding checkpoint. Although multiple fairness metrics are evaluated, we use Gini@$K$ as the primary metric for model selection, with other metrics reported for completeness. All models are implemented in PyTorch and trained on NVIDIA RTX 4080 GPUs.

\subsection{Performance Comparison}
We use BPRMF as the default base model unless otherwise specified. As shown in Table~\ref{tab:main_results}, PFA consistently achieves strong performance in both accuracy and fairness across all three datasets. Among all fairness-aware methods, PFA achieves the highest NDCG on all three datasets, indicating superior ranking quality. Compared with other methods, PFA incurs substantially smaller accuracy degradation relative to the base recommender. This indicates that the differentiable NDCG loss effectively preserves recommendation quality during fairness adaptation. Regarding provider fairness, PFA significantly improves exposure balance. It achieves the lowest Gini on all three datasets, demonstrating its effectiveness in reducing exposure inequality. Additionally, PFA obtains the highest Entropy on Amazon Beauty and Amazon Movies \& TV, and the second highest on RateBeer, reflecting more diversified and equitable exposure distribution. For the CV metric, post-processing methods such as CPFair and TaxRank achieve lower values, as they can directly manipulate exposure allocation through reranking. In contrast, PFA performs end-to-end fairness optimization and supports diverse fairness objectives by simply specifying different target distributions without algorithmic redesign. The improvements stem from HEFA's hierarchical design. Inter-group alignment reallocates exposure from head providers to mid- and tail-level providers, while intra-group alignment further distributes exposure evenly within each group. Together, these mechanisms prevent a small subset of providers from monopolizing exposure and enable fine-grained fairness control at both group and provider levels.
\begin{table}[htbp]
\centering
\caption{Performance comparison on Amazon Beauty with different base models.}
\scriptsize
\setlength\tabcolsep{2.5pt}
\renewcommand{\arraystretch}{0.9}
\begin{tabular}{llcccccc}
\toprule
\multirow{2}{*}{\textbf{Base}} & \multirow{2}{*}{\textbf{Method}} &
\multicolumn{3}{c}{\textbf{Rec. Acc.}$\uparrow$} &
\multicolumn{3}{c}{\textbf{Prov. Fair.}} \\
\cmidrule(lr){3-5} \cmidrule(lr){6-8}
& & NDCG & HR & MRR & Gini$\downarrow$ & Ent.$\uparrow$ & CV$\downarrow$ \\
\midrule
\multirow{10}{*}{{LightGCN}} 
& Base model & \textbf{0.1254} & \textbf{0.2466} & \textbf{0.1015} & 0.7102 & 6.0171 & 2.3468 \\
& Multi-FR & 0.1022 & 0.2128 & 0.0802 & 0.6747 & 6.1059 & 2.2145 \\
& Ada2Fair & 0.0987 & 0.2103 & 0.0797 & 0.6722 & 6.1132 & 2.2235 \\
& SDRO & 0.1012 & 0.2113 & 0.0805 & 0.6624 & 6.1523 & 2.2814 \\
& FairDual & 0.1150 & 0.2377 & 0.0911 & 0.6875 & 6.1582 & 2.3621 \\
& FairRec & 0.0993 & 0.2283 & 0.0683 & 0.6841 & 6.1550 & 2.3566 \\
& CPFair & 0.1111 & 0.2241 & 0.0898 & \underline{0.6529} & \underline{6.2578} & \underline{2.1598} \\
& TaxRank & 0.0982 & 0.2006 & 0.0703 & 0.6732 & 6.1851 & \textbf{2.1371} \\
& ElasticRank & 0.1020 & 0.2154 & 0.0792 & 0.6871 & 5.0856 & 4.4055 \\
& PFA [Ours] & \underline{0.1204} & \underline{0.2305} & \underline{0.0946} & \textbf{0.6428} & \textbf{6.2819} & 2.2733 \\
\midrule
\multirow{10}{*}{{SGL}}
& Base model & \textbf{0.1639} & \textbf{0.2918} & \textbf{0.1315} & 0.6973 & 6.1531 & 2.4373 \\
& Multi-FR & 0.1265 & 0.2752 & 0.1026 & 0.6702 & 6.2315 & 2.3571 \\
& Ada2Fair & 0.1087 & 0.2607 & 0.1011 & 0.6673 & 6.2147 & 2.2855 \\
& SDRO & 0.1225 & 0.2658 & 0.1153 & 0.6627 & 6.2613 & 2.2136 \\
& FairDual & 0.1439 & 0.2808 & 0.1215 & 0.6577 & 6.2431 & 2.2374 \\
& FairRec & 0.1116 & 0.2657 & 0.0947 & 0.6787 & \underline{6.3096} & 2.2882 \\
& CPFair & 0.1344 & 0.2764 & 0.1168 & 0.6562 & 6.2674 & \underline{2.1377} \\
& TaxRank & 0.1058 & 0.2685 & 0.1025 & 0.6548 & 6.2393 & \textbf{2.0003} \\
& ElasticRank & 0.1329 & 0.2795 & 0.1181 & \underline{0.6524} & 6.2789 & 2.2122 \\
& PFA [Ours] & \underline{0.1507} & \underline{0.2813} & \underline{0.1254} & \textbf{0.6386} & \textbf{6.3129} & 2.2089 \\
\bottomrule
\end{tabular}
\label{tab:backbone_results}
\end{table}

To provide an intuitive comparison, Figure~\ref{fig:scatter} illustrates the accuracy-fairness trade-off, where PFA consistently lies closest to the upper-right region across all datasets, indicating a superior balance between ranking quality and fairness. Furthermore, as reported in Table~\ref{tab:backbone_results}, PFA yields additional gains when applied to different backbones (LightGCN and SGL) on Amazon Beauty, highlighting its robustness and model-agnostic adaptability. Overall, the results demonstrate that PFA surpasses existing fairness-aware methods, achieving a more favorable trade-off between recommendation accuracy and provider exposure fairness.
\begin{figure}[htbp]
    \centering
    \includegraphics[width=\columnwidth]{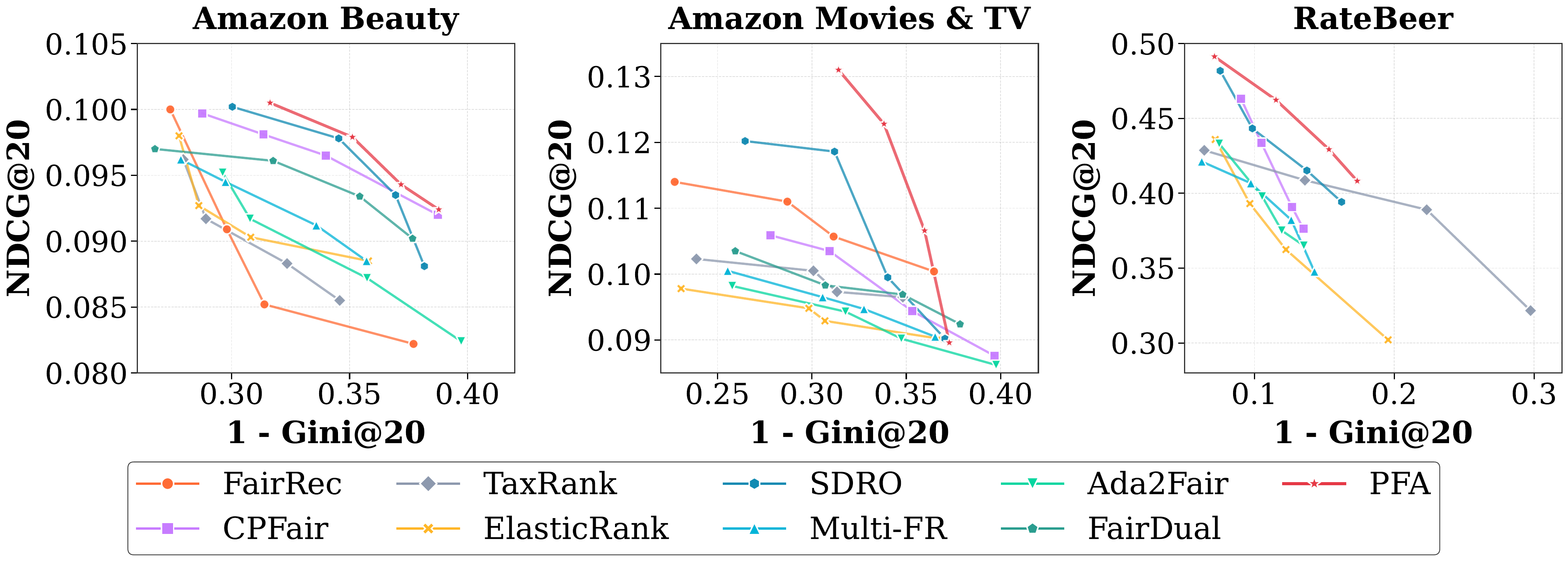}
    \caption{Accuracy-fairness trade-off comparison. Points closer to the upper-right indicate better trade-offs.}
    \label{fig:scatter}
\end{figure}

\subsection{Ablation Study}
We conduct ablation studies to examine the contribution of each component in PFA. As shown in Figure~\ref{fig:ablation}(a), we compare four settings: (1) \textbf{Base}: the pretrained recommender without any fairness intervention, (2) \textbf{$\mathcal{L}_{\text{diffNDCG}}$}: training the fairness adapter with only the differentiable NDCG loss, (3) \textbf{$\mathcal{L}_{\text{HEFA}}$}: training the fairness adapter with only the HEFA loss, and (4) \textbf{PFA}: our complete objective combining both losses. Training with only $\mathcal{L}_{\text{diffNDCG}}$ yields marginal accuracy gains but results in severe exposure imbalance, confirming that accuracy-oriented optimization alone amplifies the "rich-get-richer" effect. Applying $\mathcal{L}_{\text{HEFA}}$ alone substantially improves fairness by significantly reducing Gini, albeit at the cost of recommendation accuracy, as expected. Our full objective combines both components synergistically: $\mathcal{L}_{\text{diffNDCG}}$ serves as a regularizer that recovers accuracy while $\mathcal{L}_{\text{HEFA}}$ maintains fairness gains, achieving the best accuracy-fairness trade-off across all datasets.

Figure~\ref{fig:ablation}(b) further compares our hierarchical exposure fairness alignment ($\mathcal{L}_{\text{HEFA}}$) with direct KL divergence minimization ($\mathcal{L}_{\text{KL}}$). Since both HEFA and KL divergence align exposure distributions toward target distributions, this comparison isolates the benefit of our hierarchical decomposition. We evaluate four fairness objectives: $\mathcal{L}_{\text{KL}}$, $\mathcal{L}_{\text{diffNDCG}} + \mathcal{L}_{\text{KL}}$, $\mathcal{L}_{\text{HEFA}}$, and $\mathcal{L}_{\text{diffNDCG}} + \mathcal{L}_{\text{HEFA}}$ (PFA). The results show that HEFA-based methods consistently outperform KL-based counterparts: $\mathcal{L}_{\text{HEFA}}$ achieves comparable fairness to $\mathcal{L}_{\text{KL}}$ with higher accuracy, and PFA achieves significantly better accuracy than $\mathcal{L}_{\text{diffNDCG}} + \mathcal{L}_{\text{KL}}$ while maintaining similar fairness levels. This improvement stems from HEFA providing more targeted optimization signals by explicitly modeling hierarchical exposure dynamics across and within provider groups, rather than treating all providers homogeneously as in the global KL objective. These results validate the effectiveness of our hierarchical decomposition in achieving superior accuracy-fairness trade-offs.
\begin{figure}[!h]
    \centering
    \includegraphics[width=\columnwidth]{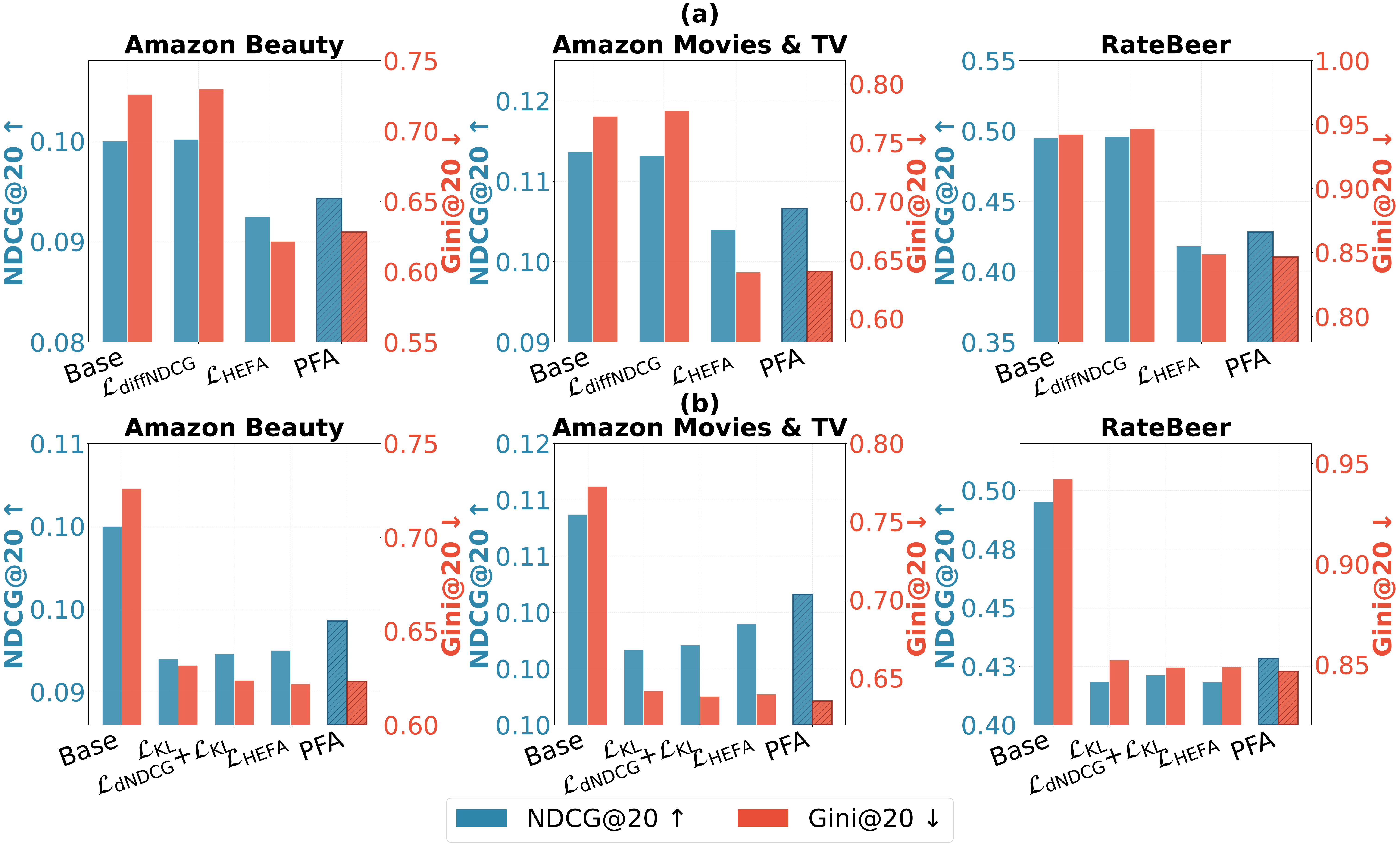}
    \caption{Ablation study across three datasets. (a) Component ablation comparing four settings: Base (no fairness intervention), $\mathcal{L}_{\text{diffNDCG}}$ only, $\mathcal{L}_{\text{HEFA}}$ only, and PFA. (b) Comparison between direct KL divergence minimization and the proposed HEFA.}
    \label{fig:ablation}
\end{figure}

\subsection{Subgroup Analysis}
To validate the effectiveness of HEFA, we conduct a subgroup analysis on providers. Specifically, we partition providers into three groups based on their historical exposure in the training set: head group (top 20\%), mid group (middle 60\%), and tail group (bottom 20\%). Fairness is evaluated from two perspectives: (1) inter-group exposure distribution, i.e., the proportion of total exposure received by each group; and (2) intra-group fairness, i.e., the uniformity of exposure allocation among providers within each group, measured by the within-group Gini. As shown in Figure~\ref{fig:subgroup_analysis} (Top), the base model exhibits severe exposure imbalance: the head group captures the vast majority of exposure, while the tail group receives a marginal share. Existing fairness methods attempt to mitigate this by redistributing exposure, but they exhibit different trade-offs. For instance, CPFair most aggressively suppresses head-group exposure, achieving the largest degree of inter-group redistribution.

\begin{figure}[!h]
    \centering
    \includegraphics[width=\columnwidth]{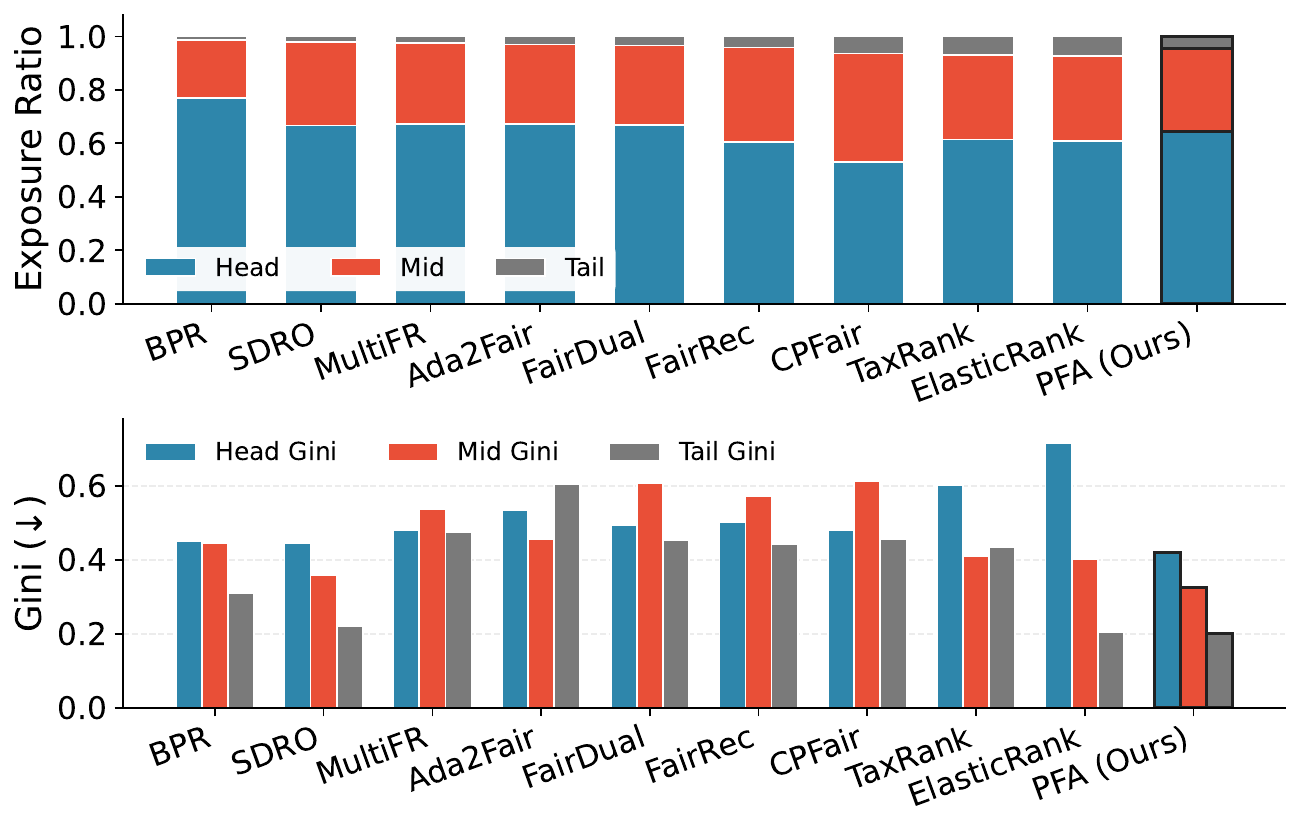}
\caption{Subgroup-level fairness comparison on Amazon Beauty. (Top) Inter-group exposure distribution. (Bottom) Intra-group fairness measured by Gini.}
    \label{fig:subgroup_analysis}
\end{figure}

More critically, as illustrated in Figure~\ref{fig:subgroup_analysis}(Bottom), most existing methods improve inter-group balance while simultaneously worsening intra-group fairness. Taking CPFair as an example, although it increases the mid and tail groups' exposure shares, the within-group Gini for both groups rises considerably. This indicates that the redistributed exposure is concentrated on a few dominant providers rather than evenly distributed, resulting in a "pseudo-fairness" phenomenon. Similar issues are observed in FairRec, FairDual, and Multi-FR. In contrast, PFA achieves genuine hierarchical fairness. While maintaining competitive inter-group distribution, PFA substantially reduces both mid and tail group Gini, achieving the lowest intra-group inequality among all methods. These results demonstrate that PFA addresses both "where to transfer exposure" and "how to distribute it within groups".

\subsection{Parameter Sensitivity Analysis}
\subsubsection{Effect of Accuracy-Fairness Trade-off Parameter}
\label{sec:Trade-off Parameter}
The hyperparameter $\lambda_{\mathrm{acc}}$ controls the weight of the accuracy loss in the overall objective function, determining the trade-off between fairness optimization and ranking quality preservation. We systematically vary it within $\{10^{-6}, 10^{-5}, 10^{-4}, 10^{-3}, 10^{-2}, 10^{-1}\}$ and examine its impact on recommendation accuracy (NDCG) and provider fairness (Gini), as shown in Figure~\ref{fig:lambda_sensitivity}. The effect of $\lambda_{\mathrm{acc}}$ can be characterized into three distinct regimes. In the \textit{fairness-dominant regime} ($\lambda_{\text{acc}} < 10^{-4}$), the fairness loss dominates optimization, achieving the lowest Gini values with moderate accuracy reduction. In the \textit{balanced transition regime} ($\lambda_{\text{acc}} \in \{10^{-4}, 10^{-3}\}$), accuracy improves while fairness degrades gradually, representing the recommended configuration range for favorable trade-offs. In the \textit{saturation regime} ($\lambda_{\text{acc}} \geq 10^{-2}$), both metrics stabilize as the accuracy loss dominates, causing the fairness adapter to approximate an identity mapping. This saturation phenomenon validates that the gradients of $\mathcal{L}_{\text{diffNDCG}}$ are inherently larger in magnitude, suppressing fairness optimization when the weighting coefficient is not properly tuned.
\begin{figure}[htbp]
    \centering
    \includegraphics[width=\columnwidth]{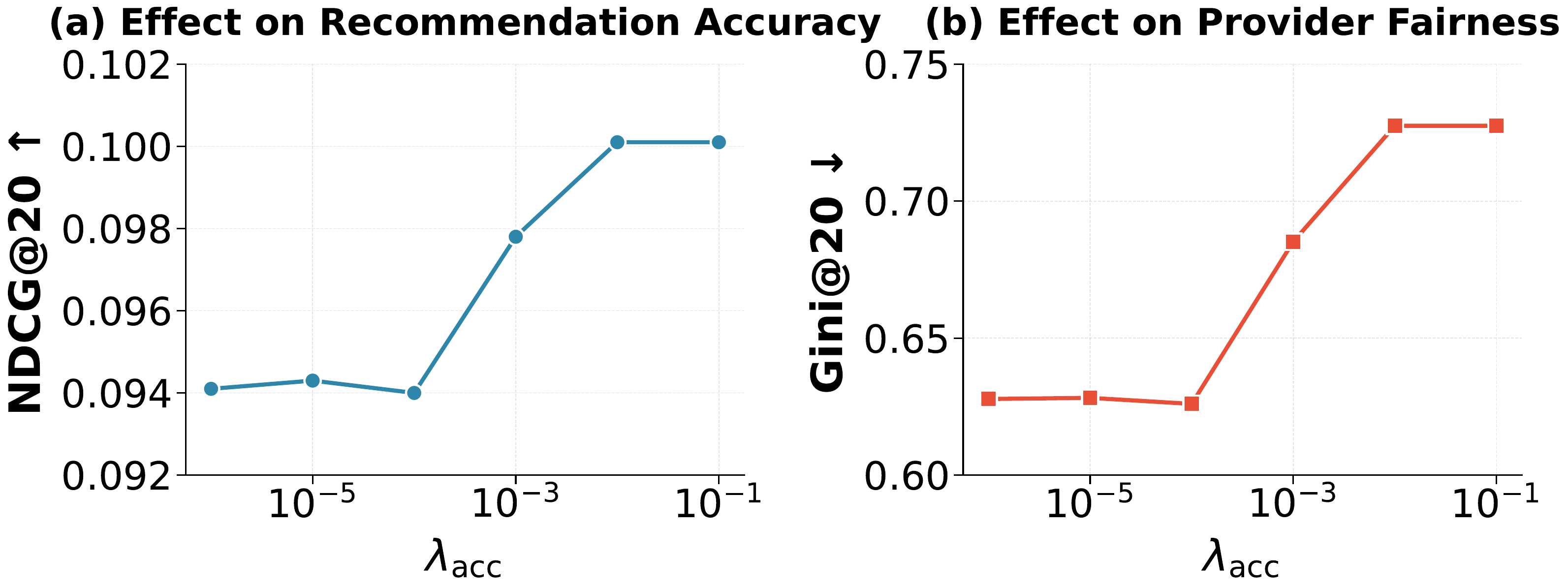}
    \caption{Sensitivity analysis of the accuracy-fairness trade-off parameter $\lambda_{\mathrm{acc}}$ on Amazon Beauty.}
    \label{fig:lambda_sensitivity}
\end{figure}

\subsubsection{Effect of Hierarchical Alignment Weights}
\begin{figure}[htbp]
    \centering
    \includegraphics[width=\columnwidth]{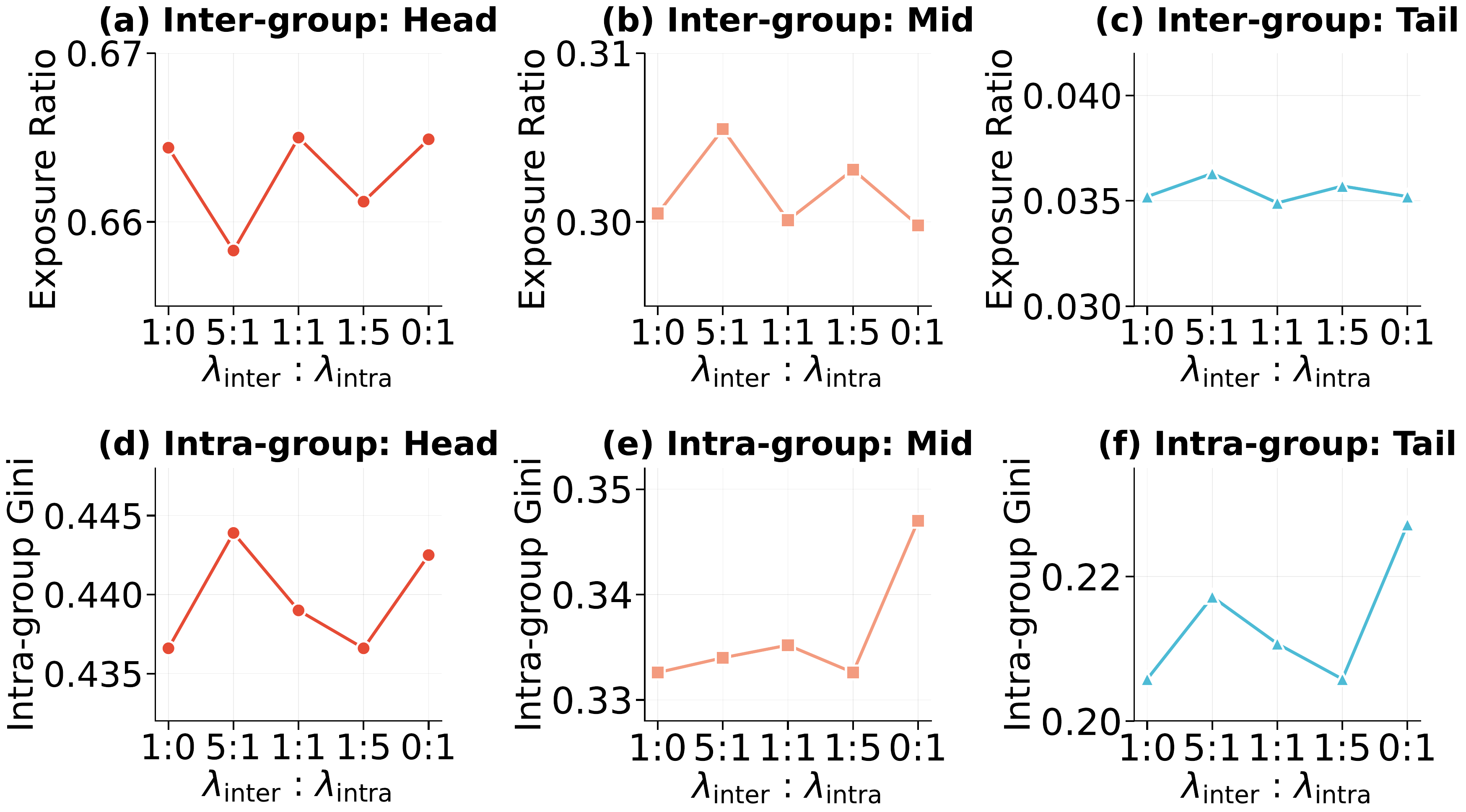}
    \caption{Sensitivity analysis of the HEFA weights $\lambda_{\mathrm{inter}} : \lambda_{\mathrm{intra}}$ on Amazon Beauty.}
    \label{fig:hierarchical_sensitivity}
\end{figure}
To evaluate the controllability of the inter-group and intra-group fairness objectives in HEFA, we fix $\lambda_{\text{acc}} = 10^{-4}$ and systematically vary the ratio of $\lambda_{\text{inter}}$ to $\lambda_{\text{intra}}$, as shown in Figure~\ref{fig:hierarchical_sensitivity}. As shown in Figures~\ref{fig:hierarchical_sensitivity}(a)-(c), increasing $\lambda_{\text{inter}}$ (from 1:1 to 5:1) reduces the Head subgroup's exposure share while increasing those of the Mid and Tail subgroups, demonstrating that $\lambda_{\text{inter}}$ effectively controls inter-group exposure redistribution. Figures~\ref{fig:hierarchical_sensitivity}(d)-(f) reveal that increasing $\lambda_{\text{intra}}$ (from 1:1 to 1:5) reduces the intra-group Gini for all subgroups, confirming that $\lambda_{\text{intra}}$ controls within-group exposure uniformity. Notably, the extreme configuration (0:1) paradoxically increases the Gini for Mid and Tail subgroups, indicating that inter-group constraints are necessary for stable optimization. These results validate that HEFA provides independent control over inter-group and intra-group fairness, allowing platforms to configure weights based on their specific fairness priorities.

\subsection{Effect of Adapter Architecture}
We investigate the impact of adapter architecture by varying the number of MLP layers and hidden dimensions. As shown in Table~\ref{tab:adapter_arch}, the 1-layer adapter achieves high accuracy but poor fairness, indicating insufficient capacity to learn effective score corrections. Increasing to 2 layers substantially improves fairness, while further increasing to 3 layers yields only marginal gains with 50\% more parameters, suggesting diminishing returns. When fixing the number of layers at 2, varying the hidden dimension has minimal impact on final performance but significantly affects convergence speed—the 64-dimensional adapter converges at epoch 19, while the 16-dimensional variant requires 46 epochs. Based on these findings, we adopt the 2-layer MLP with hidden dimension 32 as our default, which balances fairness performance, parameter efficiency, and convergence speed.
\begin{table}[htbp]
\setlength\tabcolsep{2.1pt}
\renewcommand{\arraystretch}{0.9}
\centering
\caption{Effect of adapter architecture on Amazon Beauty. The default configuration is in \textbf{bold}.}
\begin{tabular}{cccccc}
\toprule
Layers & Hidden Dim & NDCG & Gini & Params & Best Epoch \\
\midrule
1 & -- & 0.0996 & 0.6585 & 65 & 50 \\
2 & 32 & 0.0943 & 0.6282 &2,113 & 27 \\
3 & 32 & 0.0940 & 0.6210 & 3,169 & 12 \\
\midrule
2 & 16 & 0.0941 & 0.6284 & 1,057 & 46 \\
2 & 64 & 0.0944 & 0.6292 & 4,225 & 19 \\
\bottomrule
\end{tabular}
\label{tab:adapter_arch}
\end{table}

\subsection{Inference Efficiency Analysis}
We evaluate the computational efficiency of all methods by measuring total inference time on three datasets, as shown in Figure~\ref{fig:inference_time}. In-processing methods exhibit low inference times comparable to the base model, as they introduce no additional overhead during inference. However, these methods require full retraining whenever fairness objectives change, limiting their flexibility. Among post-processing methods, FairRec incurs prohibitively high costs due to its integer linear programming formulation, while CPFair and ElasticRank also scale poorly on larger datasets. As noted earlier, FairRec and FairDual fail to complete on RateBeer due to timeout. In contrast, PFA achieves inference times comparable to in-processing methods across all datasets. Although slightly slower due to the additional forward pass through the fairness adapter, PFA avoids the prohibitive overhead of optimization-based reranking, making it suitable for large-scale deployment.
\begin{figure}[htbp]
    \centering
    \includegraphics[width=\columnwidth]{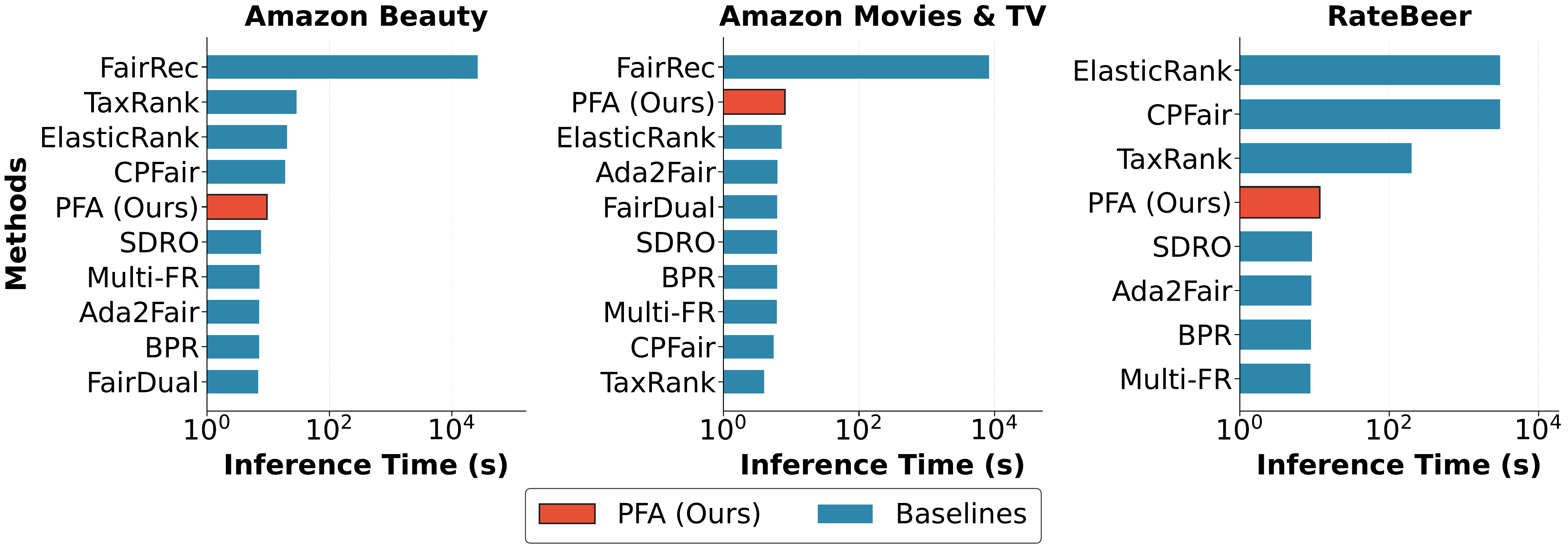}
    \caption{Inference time comparison across three datasets.}
    \label{fig:inference_time}
\end{figure}
\section{Conclusion}
In this paper, we propose Post-hoc Fairness Adaptation (PFA), a lightweight framework that decouples fairness optimization from base recommender training. PFA freezes the pretrained recommender and attaches a lightweight fairness adapter that learns personalized additive score corrections to steer recommendations toward fairer exposure distributions, avoiding retraining the backbone while overcoming the rigidity of fixed reranking strategies. To enable fine-grained fairness control, we design Hierarchical Exposure Fairness Alignment (HEFA), which decomposes fairness optimization into inter-group and intra-group components, allowing flexible adaptation to diverse fairness requirements. By jointly optimizing HEFA with a differentiable NDCG loss and leveraging a differentiable sorting network for end-to-end training, PFA effectively balances the accuracy-fairness trade-off. Extensive experiments on three public datasets demonstrate that PFA consistently outperforms existing baselines, achieving substantial fairness improvements while maintaining competitive recommendation accuracy. Future work includes exploring dynamic grouping strategies that adaptively update provider partitions based on evolving exposure patterns, and extending PFA to cold-start scenarios where new providers lack historical exposure.

\begin{acks}
This research is supported by the National Natural Science Foundation of China (No. 62272254, No. 72188101, No.U23B2031) and the New Cornerstone Science Foundation through the XPLORER PRIZE.
\end{acks}

\balance
\bibliographystyle{ACM-Reference-Format}
\bibliography{sample-base}

\end{document}